\PassOptionsToPackage{numbers,sort&compress}{natbib}
\documentclass[aps,onecolumn,notitlepage,nofootinbib,superscriptaddress,floatfix,prx,10pt]{revtex4-2}

\usepackage[utf8]{inputenc}
\usepackage[english]{babel}
\newcommand{\Fstab}{F_{\mathrm{stab}}}
\usepackage{amsmath,amssymb,amsthm,mathtools}
\usepackage[letterpaper,margin=0.68in]{geometry}

\usepackage{graphicx}
\usepackage[shortlabels]{enumitem}
\usepackage{needspace}

\usepackage{booktabs}
\usepackage{tabularx}

\usepackage{float}
\makeatletter
\let\newfloat\newfloat@ltx
\makeatother
\usepackage[ruled,vlined]{algorithm2e}

\usepackage{comment}

 \usepackage{xcolor}
\usepackage{tikz}
\usetikzlibrary{arrows.meta,positioning,calc}

\usepackage[pdfpagelabels]{hyperref}
\hypersetup{
    colorlinks=true,
    linkcolor = [rgb]{0.70,0.13,0.13},
    citecolor = [rgb]{0.13,0.55,0.13},
    urlcolor  = [rgb]{0.25,0.41,0.88}
}
\usepackage[nameinlink,noabbrev]{cleveref}

\hypersetup{
    pdftitle={Adaptivity makes single-copy stabilizer learning optimal},
    pdfauthor={ABC},
    pdfsubject={Exact and tolerant single-copy learning; tolerant memoryless testing and ordinary testing with quantum memory}
}

\theoremstyle{plain}
\newtheorem{theorem}{Theorem}[section]
\newtheorem{lemma}[theorem]{Lemma}

\theoremstyle{definition}

\crefname{theorem}{Theorem}{Theorems}
\crefname{lemma}{Lemma}{Lemmas}
\crefname{corollary}{Corollary}{Corollaries}
\crefname{proposition}{Proposition}{Propositions}
\crefname{section}{Section}{Sections}
\crefname{subsection}{Section}{Sections}
\crefname{appendix}{Appendix}{Appendices}
\Crefname{appendix}{Appendix}{Appendices}

\let\doi\relax

\newcommand{\F}{\mathbb{F}_2}
\newcommand{\E}{\mathbb{E}}

\newcommand{\ket}[1]{\lvert#1\rangle}
\newcommand{\bra}[1]{\langle#1\rvert}

\newcommand{\tr}{\operatorname{tr}}

\newcommand{\FS}{F_{\mathrm{stab}}}
\newcommand{\Stab}{\mathrm{Stab}}

\newcommand{\id}{\mathbb{I}}
\newcommand{\eps}{\varepsilon}

\newcommand{\doi}[1]{%
    \href{https://doi.org/#1}{doi:#1}%
}

\usepackage[normalem]{ulem}

\SetKwInput{KwInput}{Input}
\SetKwInput{KwOutput}{Output}

\begin{document}
\raggedbottom

\title{Adaptivity is all you need:\\ Optimal stabilizer learning using just single-copy measurements}

\author{Lennart Bittel$^{\dagger}$}
\email{l.bittel@fu-berlin.de}
\affiliation{Dahlem Center for Complex Quantum Systems,
Freie Universit{\"a}t Berlin, 14195 Berlin, Germany}

\author{Jens Eisert$^{\dagger}$}
\email{jense@zedat.fu-berlin.de}
\affiliation{Dahlem Center for Complex Quantum Systems,
Freie Universit{\"a}t Berlin, 14195 Berlin, Germany}
\affiliation{Helmholtz-Zentrum Berlin f{\"u}r Materialien und Energie,
Berlin, Germany}

\author{Weiyuan Gong$^{\dagger}$}
\email{wgong@g.harvard.edu}
\affiliation{John A. Paulson School of Engineering and Applied Sciences,
Harvard University, Allston, Massachusetts 02134, USA}

\author{\\Antonio Anna Mele$^{\dagger}$}
\email{a.mele@fu-berlin.de}
\affiliation{Dahlem Center for Complex Quantum Systems,
Freie Universit{\"a}t Berlin, 14195 Berlin, Germany}

\author{Louis Schatzki$^{\dagger}$}
\email{l.schatzki@ibm.com}
\affiliation{Dahlem Center for Complex Quantum Systems,
Freie Universit{\"a}t Berlin, 14195 Berlin, Germany}


\date{\today}

\begin{abstract}
Stabilizer states are central to quantum computing, underlying quantum error correction, benchmarking, and efficient classical simulation. Yet their learnability exhibits a striking gap: an \(n\)-qubit stabilizer state can be learned from \(\Theta(n)\) copies using two-copy Bell measurements, whereas non-adaptive single-copy measurements require \(\Omega(n^2)\) copies. Here we show that adaptivity completely closes this gap. We give a polynomial-time adaptive algorithm that learns an arbitrary \(n\)-qubit stabilizer state from \(\Theta(n)\) single-copy Clifford measurements, matching the optimal sample complexity of Bell sampling without any multi-copy measurements. The same ideas yield a sample-optimal single-copy tolerant tester and, with \(k\) qubits of quantum memory, the optimal testing tradeoff \(\Theta(n-k+1/\varepsilon)\) at infidelity \(\varepsilon\). Finally, we show that this adaptive mechanism extends beyond exact stabilizer states: states of stabilizer nullity at most \(r\), including states prepared by Clifford circuits with a bounded number of \(T\) gates, can be learned using \(O(n2^{r})\) single-copy measurements.
\end{abstract}
\maketitle

\begingroup
\renewcommand{\thefootnote}{\fnsymbol{footnote}}
\footnotetext[2]{Authors are listed in alphabetical order.}
\endgroup

\vspace*{-0.7cm}
\section{Introduction}
Stabilizer states constitute one of the most fundamental and extensively studied classes of quantum states and take center stage in quantum information science. The stabilizer formalism underpins both quantum error correction and fault-tolerant quantum computation~\cite{Got97}, as well as provides efficient classical descriptions and simulation algorithms~\cite{AG04}. At the same time, stabilizer states are central objects in quantum verification, certification, and benchmarking~\cite{FL11,MGE11,HKP20,BenchmarkingReview,PRXQuantum.2.010201}. This raises 
the basic operational question of how reliably and resource-efficiently an unknown stabilizer state can be identified from experimental data. The problem has motivated substantial recent interest in learning and testing stabilizer states under different restrictions on the available measurements and quantum memory~\cite{AG08,Mon17,ABDY23,GIKL24,CLL24,AD24,BvDH25,HH25,AS26}.


The problem of \emph{learning stabilizer states} asks, given independent copies of an unknown stabilizer state, to recover a classical description of the state. Montanaro~\cite{Mon17} gave a polynomial-time algorithm achieving the optimal sample complexity $\Theta(n)$ via Bell sampling, which is intrinsically a two-copy measurement: corresponding qubits of two copies are jointly measured in the Bell basis. Repeating this procedure $O(n)$ times reveals enough information to reconstruct the stabilizer group, and so the recovering the stabilizer state. Since $\Omega(n)$ copies are necessary information-theoretically~\cite{AG08}, this achieves the optimal sample complexity. This optimal scaling, however, comes at the price of requiring coherent access to two copies of the unknown state at the same time.

It is therefore highly natural, both conceptually and experimentally, to ask what can be achieved using only single-copy measurements. Such a class of protocols processes one copy at a time and requires no coherent quantum memory between different copies. Avoiding this requirement can be significantly more convenient in practice, and single-copy and memory-limited approaches to quantum learning and testing have consequently received increasing attention~\cite{ABDY23,CLL24,HH25,AS26}. 

Known bounds suggested that restricting to single-copy measurements might incur a substantial cost for stabilizer learning. Ref.~\cite{ABDY23} showed that non-adaptive single-copy measurements require $\Omega(n^2)$ copies already for quadratic phase states, a subclass of stabilizer states. Thus, there is a quadratic separation between the $\Theta(n)$ copies sufficient with Bell sampling and the $\Omega(n^2)$ copies required by non-adaptive single-copy strategies. Crucially, however, this lower bound does not apply to adaptive protocols: it assumes that the measurements performed on different copies are chosen independently of previous measurement outcomes. 
This leaves open the following basic question:
\begin{quote}
\emph{Are adaptive single-copy measurements already sufficient for sample-optimal stabilizer learning?}
\end{quote}

In this work, we answer this question in the affirmative. We give a time-efficient algorithm that learns an arbitrary \(n\)-qubit stabilizer state using \(\Theta(n)\) single-copy adaptive measurements, matching the optimal sample complexity achievable by Bell sampling, while requiring no coherent access to multiple copies. The algorithm uses only Clifford operations followed by computational-basis measurements. Thus, classical adaptivity alone suffices for sample-optimal stabilizer learning.

The main idea behind the algorithm is simple. If all stabilizers of the state were \(Z\)-type, then the state would already be a computational-basis state and could be identified immediately. The goal is therefore to progressively convert non-diagonal stabilizers into \(Z\)-type stabilizers, without destroying those obtained in previous rounds.

To do so, in each round we measure two copies separately in the computational basis and take the bitwise difference of the two outcomes. For a stabilizer state, this difference identifies the $X$-part of a stabilizer that is not diagonal in the computational basis. We then use this information to adaptively choose a Clifford operation that turns this stabilizer into a $Z$-type stabilizer, while preserving the $Z$-type stabilizers found in previous rounds. Repeating this procedure monotonically increases the number of known $Z$-type stabilizers, until the state is transformed into a computational-basis state and is therefore completely determined. Keeping track of the Clifford operations applied along the way then recovers the original stabilizer state. 

We also consider the \emph{stabilizer testing} problem: distinguish stabilizer states from states whose
fidelity with every stabilizer state is at most $1-\varepsilon$. Combining the known $\Omega(n)$
single-copy lower bound~\cite{AS26} with the $\Omega(1/\varepsilon)$
cost of state discrimination gives a lower bound of
$\Omega(n+1/\varepsilon)$ copies. On the other hand, the best previous single-copy
upper bound for general $\varepsilon$ was
$O(n/\varepsilon)$~\cite{AS26}. Our adaptive learning procedure
yields a polynomial-time single-copy tester with optimal sample
complexity $\Theta(n+1/\varepsilon)$, therefore closing the gap. Further, our tester works even for mixed inputs.

We then extend our approach to the tolerant setting. We first show
that a slight modification of our exact stabilizer learning protocol
also gives a \emph{tolerant learner}, making the algorithm robust to
deviations from the stabilizer-state promise. More precisely, let
$\Fstab(\rho)=\max_S\langle S|\rho|S\rangle$ denote the fidelity with
a closest pure stabilizer state, and assume
$\Fstab(\rho)>\sqrt{2/3}$. With high probability, our protocol then finds a stabilizer state whose
fidelity with $\rho$ is within any desired additive error of the
optimum $\Fstab(\rho)$, even when $\rho$ is mixed.

This also allows us to perform \emph{tolerant stabilizer testing}:
given $\varepsilon_1<\varepsilon_2$, distinguish
$\Fstab(\rho)\geq1-\varepsilon_1$ from
$\Fstab(\rho)\leq1-\varepsilon_2$.
This problem has been studied in several recent
works~\cite{GIKL24,AD24,BvDH25}. For pure inputs and suitable constant
thresholds, two-copy Bell measurements suffice to test using $O(1)$
copies, independently of the number of qubits~\cite{GNW21,GIKL24}.
Whether efficient tolerant testing is possible using only single-copy
measurements, however, remained open~\cite{HH25}.
We show that our adaptive approach also solves this problem:
in the parameter regime specified below, it gives a polynomial-time
single-copy tester with optimal sample complexity
$\Theta\!\left(n+\varepsilon_2/(\varepsilon_2-\varepsilon_1)^2\right)$.

Furthermore, we study stabilizer testing when the protocol is allowed to retain $k$ qubits of coherent quantum memory between successive copies. For constant distance parameter, this model interpolates between the single-copy setting, $k=0$, where the optimal sample complexity is $\Theta(n)$~\cite{HH25,AS26}, and the regime in which an entire copy can be stored coherently, where $\Theta(1)$ copies suffice~\cite{GNW21}. Building on these results and recent work on stabilizer testing with limited quantum memory~\cite{AS26}, we determine the optimal memory--sample tradeoff. In particular, for non-tolerant stabilizer testing with distance parameter $\varepsilon$, the optimal sample complexity is $\Theta(n-k+1/\varepsilon)$, even with mixed inputs. Notably, learning and testing stabilizer states both requires $\Theta(n)$ copies until $k = n-o(n)$. That is, to achieve sub-linear testing, one must be able to store nearly the entire state.

Finally, beyond exact stabilizer states, the same mechanism remains effective when learning pure states with stabilizer nullity $r$, meaning that its unsigned Pauli stabilizer space has dimension $n-r$. 
Recall that in the stabilizer state case, every computational-basis difference is the $X$-part of a stabilizer. 
At nullity $r$, only a $2^{-O(r)}$ fraction of the differences need have this property. 
The crucial observation is that the remaining differences are nevertheless harmless: they still lie in the orthogonal complement of the diagonal stabilizer space, and therefore the same Clifford update preserves every diagonal stabilizer already present. 
Consequently, the adaptive compression is slowed down only by a factor $2^{O(r)}$.
After $O(2^rn)$ rounds of adaptivity, the transformed state is supported on an affine subspace of dimension at most $r$.
Learning this subspace up to high probability mass and performing pure-state tomography using random Clifford measurements only on the remaining logical subsystem gives $O(n2^r)$ single-copy measurements at constant accuracy and confidence.
In particular, since a $t$-doped state with $t$ non-stabilizer ($T$) gates has stabilizer nullity at most $2t$ ($t$), $t$-doped stabilizer states can be learned using $n2^{O(t)}$ copies at constant accuracy and confidence, improving over the result in Ref.~\cite{CLL24}.

Taken together, our results introduce a single-copy adaptive subroutine
for stabilizer learning and use it to close several gaps in quantum
learning and testing. More broadly, our work contributes to shedding light on when
adaptivity helps in quantum learning. Its benefits are highly
task-dependent, and provable separations from non-adaptive strategies
remain relatively rare~\cite{CGYZ25} to date. Known examples where single-copy adaptivity helps
include quantum state tomography under the infidelity
metric~\cite{Mahler13,CHLLS23} and bosonic Gaussian tomography, where
it can remove nearly all dependence on the state energy~\cite{BMEM25,CGYZ26}.
In contrast, for general mixed-state tomography in trace distance,
adaptive single-copy measurements offer no asymptotic sample-complexity
advantage over non-adaptive ones~\cite{CHLLS23}. Our results show that
stabilizer learning is a natural structured setting in which adaptivity makes a
fundamental difference. The quadratic gap between non-adaptive
single-copy learning and two-copy Bell sampling can be closed by
classical feedback alone: single-copy measurements already suffice
for sample-optimal learning, without coherent access to multiple copies.

\subsubsection{Conventions}
We write $\psi=\ket\psi\bra\psi$ when a pure state is used as an operator, and equalities between states are understood up to global phase. A Pauli label $(a,b)\in\F^{2n}$ represents $X^aZ^b$ up to phase; two Pauli labels commute when $a\cdot b'+b\cdot a'=0$. The unsigned stabilizer space consists of the labels whose Paulis have the state as an eigenvector. We use $\FS(\rho):=\max_{\ket S\in\Stab_n}\bra S\rho\ket S$ for stabilizer fidelity, where $\Stab_n$ is the set of pure $n$-qubit stabilizer states. Testing distances are infidelities $1-\FS$, whereas tomography accuracy is measured in trace distance $\frac12\|\rho-\sigma\|_1$.  

\section{Exact learning}
In this section we present a polynomial-time algorithm~\ref{alg:stab_learn} that learns
an unknown stabilizer state using only $\Theta(n)$ single-copy adaptive measurements.

The central idea of the protocol is to construct a Clifford circuit $C$ such that $C\ket{\psi} = \ket{z}$, where $\ket{\psi}$ is the unknown stabilizer state and $\ket{z}$ is a computational-basis state. Once $C$ is constructed, measuring $C\ket{\psi}$ in the computational basis reveals $z$, and then the state $\ket{\psi} = C^\dagger\ket{z}$ can be reconstructed. Unless the transformed state is already a computational-basis state, each round increases the dimension of its diagonal stabilizer subspace with probability at least $1/4$, and never decreases this dimension. More formally, a state is a computational-basis state if and only if its unsigned stabilizer space equals $\mathcal{Z} \coloneqq \{ (0,z) | z\in \F^n \}$. That is, the state is a joint eigenvector of all $Z$-type Pauli strings. Consider a stabilizer state $\psi$ with unsigned stabilizer space $L_\psi$. Define $D(\psi)= \{s \in \F^n \vert (0,s) \in L_\psi\}$, which is the subspace of unsigned $Z$-type stabilizer labels. Our learning algorithm works by randomly increasing the
size of $D$ in each step and never decreasing it. Once $D(C\psi C^\dagger) = \F^n$ we are done and can read off the computational-basis state $\ket{z}$ to complete the learning procedure.  

The protocol is given explicitly in Algorithm~\ref{alg:stab_learn}. To illustrate an example, consider the Bell state
$(\ket{00}+\ket{11})/\sqrt2$, with stabilizer generators $XX$ and
$ZZ$. Initially, there is one independent diagonal stabilizer,
$ZZ$, which requires the two bits in each computational-basis
outcome to agree. Measuring two copies separately gives
$x,y\in\{00,11\}$. If the outcomes differ, their bitwise difference
is $u=x+y=11$. A CNOT from the first qubit to the second maps
the state to $\ket+\ket0$ and the generators to $XI$ and $IZ$.
The diagonal constraint now involves only the second qubit, so
we can act on the first without disturbing it. If the algorithm
chooses $H$ rather than $HS$ as the local correction, $XI$ becomes
$ZI$ while $IZ$ remains unchanged. The state is then $\ket{00}$ (and the diagonal stabilizer dimension has increased from one to two).

\begin{theorem}[Main theorem]
\label{thm:explicit-exact}
For $n\geq1$ and $0<\delta\leq1/3$,
Algorithm~\ref{alg:stab_learn} exactly learns an unknown pure
$n$-qubit stabilizer state with probability at least $1-\delta$,
in polynomial time, using
\[
N=O\!\left(n+\log\frac1\delta\right)
\]
copies and only adaptive single-copy measurements.
This sample complexity is optimal up to universal constant factors,
even allowing collective measurements.
More precisely, the algorithm runs for $T=2n+O\!\left(\sqrt{n\log(1/\delta)}+\log(1/\delta)\right)$ rounds, uses $N=2T+1$ copies, and takes $O(nT^2)$ time.
\end{theorem}

\begin{proof}
At the beginning of a round, let $\ket{\phi}=C\ket{\psi}$ denote the current transformed state. Let $L_\phi$ be its unsigned stabilizer space and define $D(\phi):=\{s\in\F^n:(0,s)\in L_\phi\}$, the subspace of $Z$-type stabilizer labels. Let $\ell:=\dim D(\phi)=\log_2|D(\phi)|$.
The algorithm first measures two copies of $\ket{\phi}$ in the computational basis and adds the outcomes, a procedure known as computational difference sampling~\cite{HH25,GIKL25}. The resulting vector $u:=x+y$ is uniformly distributed over $D(\phi)^\perp$, the orthogonal complement of $D(\phi)$ with respect to the standard inner product on $\F^n$. Moreover, $D(\phi)^\perp$ is precisely the projection of $L_\phi$ onto its $X$-coordinates. Hence, for every sampled $u\in D(\phi)^\perp$, there exists some $v\in\F^n$ such that $(u,v)\in L_\phi$. If $u\neq0$, the algorithm attempts to update $C$ so that such a stabilizer is mapped to a new $Z$-type stabilizer, thereby increasing $\ell$ by one.

Say that $u\neq 0$. Fix an index $i$ such that $u_i=1$; for concreteness, we take $i$ to be the first such index (the pivot). Consider the CNOT circuit
\begin{equation}
    F_u := \prod_{\substack{j\neq i\\u_j=1}} \mathrm{CNOT}_{i\rightarrow j}\ .
\end{equation}
One can verify that this Clifford maps a stabilizer label $(u,v)$ to $(e_i,v')$ for some $v'\in\F^n$. Moreover, CNOT gates map $X$ strings to $X$ strings and $Z$ strings to $Z$ strings, so $\ell$ is preserved under $F_u$.

It remains to understand the action of $F_u$ on the pre-existing diagonal stabilizers. For $s\in D(\phi)$, the transformed label is
\begin{equation}
    (0,s)\overset{F_u}{\longmapsto}(0,s'),
    \qquad
    s'=s+e_i\langle u+e_i,s\rangle.
\end{equation}
Its pivot coordinate satisfies $s'_i=s_i+\langle u+e_i,s\rangle=\langle u,s\rangle=0$, where the last equality follows from $u\in D(\phi)^\perp$. Hence every pre-existing $Z$-type stabilizer acts trivially on qubit $i$ after applying $F_u$, and is therefore unaffected by any subsequent local Clifford acting only on that qubit.

The stabilizer $(e_i,v')$ acts on qubit $i$ as either $X_i$ or $Y_i$. Accordingly, one of the two local corrections $H_i$ and $H_iS_i$ maps it to a $Z$-type stabilizer. Since the algorithm chooses uniformly between these two corrections, conditioned on $u\neq0$, it increases $\ell$ by one with probability $1/2$\footnote{The increase cannot be more than one because any new $Z$-type stabilizer must come from one of the form $(e_i, w)$ and hence is equal to $(e_i,v')$, the stabilizer we are implicitly converting, plus an element of $D$.}. Moreover, as shown above, both corrections preserve all pre-existing $Z$-type stabilizers, so $\ell$ can never decrease.

We now bound the number of rounds needed to increase $\ell$ to $n$. Let $r:=n-\ell$ be the codimension of $D(C\psi C^\dagger)$. Computational difference sampling yields a non-zero $u$ with probability $\Pr[u\neq0]=1-2^{-r}$. Conditioned on $u\neq0$, a correction that increases $\ell$ is chosen with probability $1/2$. Hence, whenever the current value is $r>0$, the probability of decreasing $r$ by at least one is
\begin{align}
    p_r:=\frac{1-2^{-r}}{2}\ .
\end{align}
Let $\tau$ be the number of rounds until $r = 0$. Since an update can never increase $r$ and the process is a Markov chain, the number of rounds to go from $r$ to $r-1$ has distribution $\operatorname{Geom}(p_r)$. For each possible value of $r$ define an independent random variable $G_r\sim\operatorname{Geom}(p_r)$. To capture the worst case of starting from $r=n$, we define
\begin{align}
    T_{\mathrm{stop}}:=\sum_{r=1}^n G_r.
\end{align}
For every $t\geq0$, we have $\Pr[\tau>t]\leq\Pr[T_{\mathrm{stop}}>t]$. Moreover,
\begin{align}
    \E[T_{\mathrm{stop}}]
    &=
    \sum_{r=1}^n\frac{1}{p_r}
    =
    \sum_{r=1}^n\frac{2}{1-2^{-r}}
    <2n+4\ .
\end{align}
Applying the upper-tail bound for sums of independent geometric random variables from Ref.~\cite[Theorem~2.1]{Jan18}, we obtain
\begin{align}
    \Pr\left[
        T_{\mathrm{stop}}
        >
        2n+K\left(
            \sqrt{n\log(1/\delta)}+\log(1/\delta)
        \right)
    \right]
    \leq\delta\ ,
\end{align}
where $K$ is a universal constant. By the preceding tail comparison, with probability at least $1-\delta$ it reaches a Clifford $C$ such that $C\ket{\psi}=\ket{z}$ within
$2n+O(\sqrt{n\log(1/\delta)}+\log(1/\delta))$ rounds. Once $r=0$, every later difference is zero and the circuit remains unchanged. Explicitly, $T=\lceil 2n+4+\sqrt{8(2n+4)\log(1/\delta)}+8\log(1/\delta)\rceil$ suffice. A single additional computational-basis measurement then determines $z$.
Once we have obtained $C$ and $z$, the pair $(C,z)$ already gives an efficient classical description of the learned state as $\ket{\psi}=C^\dagger\ket{z}$ (whose stabilizer group can be computed efficiently if needed).
The matching $\Omega(n+\log(1/\delta))$ lower bound, even allowing arbitrary collective measurements, follows from Ref.~\cite{chang2026sampleoptimallearningstabilizerstates,Mon17}.
\end{proof}

\begin{algorithm}[t]
\caption{Single-copy stabilizer learning}
\label{alg:stab_learn}
\KwInput{$2T+1$ independent copies of an unknown stabilizer state $\ket{\psi}$, with $T=2n+O\!\left(\sqrt{n\log(1/\delta)}+\log(1/\delta)\right)$.}
$C \leftarrow \id$\;
\For{$t = 1$ to $T$}{
    Apply $C$ to two independent copies, one at a time, and measure them in the computational basis to obtain $x,y \in \F^n$\;
    $u \leftarrow x+y\in \F^n$\;
    If $u = 0$, do nothing and proceed to the next iteration\;
    Otherwise, $C\leftarrow F_{u}C$, where $F_{u} = \prod_{\substack{j \neq i\\ u_j = 1}}\text{CNOT}_{i\rightarrow j}$ and $i$ is the first index such that $u_i=1$\;
    Draw an independent uniform bit $b$ and set $C\leftarrow H_iS_i^bC$\;
}
Measure $C\ket{\psi}$ once in the computational basis to obtain $z\in\F^n$\;
\Return{the classical pair $(C,z)$ describing $C^\dagger\ket{z}$}
\end{algorithm}

This ends the description of the algorithm and of its performance guarantees. While we have focused on qubit systems for clarity, the same argument generalizes to systems of prime local dimension, as explained in Appendix~\ref{sec:prime-dimensional-qudits}.

\section{Tolerant learning}
\label{sec:tolerant-learning}

We now extend the learner to states that are close to a stabilizer, but
need not themselves be stabilizer states or even pure. We run the same
adaptive Clifford updates and save one stabilizer candidate in every
round. We then estimate the candidates' fidelities on fresh copies and
return the best one. The main task is to show that the list contains a
closest stabilizer with high probability. We prove this using the
expected absorption time of an auxiliary birth--death chain as a
potential, followed by a standard drift concentration theorem.

\begin{theorem}[Tolerant single-copy learning]
\label{thm:tolerant}
Let $n\geq1$, $0<\xi\leq1/2$, and $0<\delta\leq1/3$. Suppose that $\FS(\rho)\geq f_0 > \sqrt{2/3}$ for a known $f_0$, where
$\rho$ is an $n$-qubit density matrix, possibly mixed. Set
$\mu:=3f_0^2/2-1>0$. There is an adaptive single-copy algorithm using
$O(R_\mu(n,\delta)+\xi^{-2}\log(R_\mu(n,\delta)/\delta))$ copies that
outputs a stabilizer $\ket{\widehat S}$ satisfying
$\bra{\widehat S}\rho\ket{\widehat S}\geq\FS(\rho)-\xi$ with probability
at least $1-\delta$. The round budget can be taken to be
\begin{align}
R_\mu(n,\delta)
&:=2\left\lceil\frac{13.5(n+1)}{\mu}\right\rceil
+\left\lceil\frac{1458}{\mu^2}\log\frac{2}{\delta}\right\rceil
=O\!\left(\frac n\mu+\frac1{\mu^2}\log\frac1\delta\right).
\label{eq:tolerant-round-budget}
\end{align}
The algorithm uses only Clifford operations and computational-basis
measurements and runs in time $O(n^2R_\mu^2\xi^{-2}\log(R_\mu/\delta))$.
\end{theorem}

\begin{proof}
Fix a closest stabilizer $\ket{S_\star}$ and write
$f:=\bra{S_\star}\rho\ket{S_\star}=\FS(\rho)\geq f_0$.
Starting from $C_0=I$, run $R:=R_\mu(n,\delta)$ adaptive rounds.
In round $t=0,\ldots,R-1$, measure two fresh copies of
$\rho_t:=C_t\rho C_t^\dagger$ in the computational basis, obtaining
$x_t,y_t$. Before changing $C_t$, record
$\ket{S_t}:=C_t^\dagger\ket{x_t}$; then use $x_t+y_t$ for the update
in Algorithm~\ref{alg:stab_learn}. Each round uses two copies and
produces one stabilizer candidate at no additional cost.

We will track the evolution of the reference state $\ket{S_*}$ under the learned Clifford. Let $\ket{\phi_t} := C_t\ket{S_*}$, $D_t := D(\phi_t)$, and $r_t := n-\dim D_t$. The quantities are not known to the algorithm and, unlike the case of $\rho$ being a stabilizer state, $r_t$ could either increase by $1$, decrease by $1$, or stay the same after an update to $C$. When $r_t = 0$, $\ket{\phi_t}$ is a computational basis state. Hence, measuring $\rho_t$ in the computational basis yields $\ket{\phi_t}$ with probability
\begin{align}
    \langle \phi_t \vert \rho_t \vert \phi_t\rangle = \langle S_*\vert \rho \vert S_*\rangle \geq f_0\ .
\end{align}

We are interested in the number of rounds required for both $r_t=0$ and for $\ket{\phi_t}$ to then be observed. To analyze this, we consider a random walk with state space $r\in \{-1,0,1,\ldots, n\}$, where $r_t = -1$ represents properly observing $\phi_t$ when $D_t = \F^n$. Clearly, adding this additional state can only increase the hitting time and we will continue with this state space. 

Next, we bound the transition probabilities. To do so, first recall that any stabilizer state has a uniformly computational-basis distribution on an affine subspace $A_t = z_t + D_t^\perp$, for some $z_t \in \F^n$. Since
$\Pi_{A_t}:=\sum_{x\in A_t}\ket{x}\bra{x}\geq\ket{\phi_t}\bra{\phi_t}$,
a measurement of $\rho_t$ lies in $A_t$ with probability at least $f$. If both outcomes lie there, their difference belongs to $D_t^\perp$,
so the preservation argument of Theorem~\ref{thm:explicit-exact}
applies to the reference state. Therefore, writing $B:=1-f_0^2$,
\begin{align}
\Pr[r_{t+1}=r_t+1\mid C_t]
&\leq1-\tr(\rho_t\Pi_{A_t})^2
\leq1-f^2\leq B.
\label{eq:tolerant-backward}
\end{align}
Note that this bound holds even in the case of $r_t = 0$ since the probability of not observing the correct computational basis state is at most $1-f \leq 1-f^2 \leq 1-f_0^2$. 

To bound the probability of $r$ decreasing, suppose $r_t=r\geq1$, set
$K:=|A_t|=2^r$, and write $p_x:=\bra{x}\rho_t\ket{x}$ and
$c_t:=\sum_{x,y\in A_t,\,x\neq y}p_xp_y$.
Here $c_t$ is the probability of two distinct outcomes in $A_t$.
The flat reference amplitudes, positivity, and Cauchy--Schwarz give
\begin{align}
f
&\leq\frac1K\left(
\sum_{x\in A_t}p_x+
\sum_{\substack{x,y\in A_t\\x\neq y}}\sqrt{p_xp_y}
\right)
\leq\frac{1+\sqrt{K(K-1)c_t}}{K}.
\label{eq:tolerant-distinct-outcomes}
\end{align}
Since $Kf>1$, we obtain $c_t\geq(Kf-1)^2/[K(K-1)]$.
Two distinct outcomes in $A_t$ yield a nonzero reference stabilizer
$X$-label. The independent random local correction then creates a new
diagonal stabilizer with probability at least $1/2$. Consequently, with $A:=f_0^2/2$,
\begin{align}
\Pr[r_{t+1}=r_t-1\mid C_t]
&\geq\frac{c_t}{2}
\geq\frac{(Kf-1)^2}{2K(K-1)}
=\frac{f^2}{2}-\frac1{2K}+\frac{(1-f)^2}{2(K-1)}
\geq a_r:=A-2^{-r-1}.
\label{eq:tolerant-forward}
\end{align}
If $r_t = 0$, the downwards transition probability is $f \geq f_0 =: a_0$.

For the sake of analysis, consider the auxiliary time-homogeneous
birth--death chain $(X_t)_{t\in\mathbb N}$ with state space
$\{-1,0,1,\ldots,n\}$, where $-1$ is absorbing. For $1\leq r<n$,
its downward transition probability is $a_r$, its upward transition
probability is $B$, and the remaining probability is assigned to
staying at $r$. At $r=n$ there is no upward transition. From $r=0$,
the chain moves to $-1$ with probability $f_0$ and to $1$ with
probability $1-f_0$.
Let
\begin{align}
    T:=\inf\{t\in\mathbb N_+:X_t=-1\}
\end{align}
be its hitting time, and we will assume that $X_0 = n$ unless otherwise stated. Let $\tau$ denote the actual
number of rounds until $\ket{S_\star}$ is first recorded. The comparison
chain is chosen pessimistically, so its hitting time stochastically
dominates the actual hitting time:
\begin{align}
    \Pr[T\geq k]\geq\Pr[\tau\geq k]
\end{align}
for every $k\in\mathbb N_+$. We therefore continue by analyzing $T$
using tools from drift theory.

 We would like to use a tail bound on $T$ assuming that the drift, $\E[X_t - X_{t+1}\vert X_t] = a_r-B$, is positive. However, the drift is $a_r-B = \mu - 2^{-r-1}$, which may be negative for $r$ close enough to $0$. To ameliorate this, we will instead consider the mapping~\cite{he2004study}
 \begin{align}
     g(x) = \E[T\vert X_0 = x], & \quad \E[g(X_t)-g(X_{t+1})
\mid g(X_0),\ldots,g(X_t),T>t]=1.\ .
 \end{align}
 This mapping takes the random process $(X_t)_{t \in \mathbb{N}_+}$ to a new process $(Y_t := g(X_t))_{t\in \mathbb{N}_+}$ such that the drift is always exactly $1$ (until the chain hits the absorbing state). The new state space is some subset of $\mathbb{R}_+$, $Y_t =0$ if and only if $X_t = -1$, and the hitting time for this chain to reach $Y = 0$ is exactly $T$, the hitting time of the auxiliary chain.

 To apply a drift concentration bound, it remains to bound the change in $Y_t$ per round. That is, we want a bound $\vert Y_t - Y_{t+1}\vert \leq \gamma$ for all times. Since $\vert X_t-X_{t+1}\vert\leq1$, it suffices to bound
$d_r:=g(r)-g(r-1)$ for all $0 < r \leq n$. These expected hitting times satisfy
 \begin{align}
     g(r) = 1+a_r g(r-1) + Bg(r+1) + (1-a_r-B)g(r)\ .
 \end{align}
 This then yields the recurrence relations
\begin{align}
d_r & = \frac{1}{a_r}+\frac{B}{a_r}d_{r+1} \quad(1\leq r<n),
& a_nd_n&=1,
& f_0d_0-(1-f_0)d_1&=1.
\label{eq:tolerant-H-recurrence}
\end{align}
This recurrence relation has solutions, for $1\leq r \leq n$,
\begin{align}
d_r
&=\sum_{k=r}^n\frac{B^{k-r}}{a_ra_{r+1}\cdots a_k}=\frac1A\sum_{k=r}^n
\left(\frac BA\right)^{k-r}
\prod_{j=r}^k
\left(1-\frac{2^{-j}}{f_0^2}\right)^{-1}.
\label{eq:tolerant-H-increments}
\end{align}

The product term can be uniformly bounded: for $j=1$ the factor is at most $4$, while for $j\geq 2$ we have that $\ln (1-2^{-j}/f_0^2)^{-1} \leq (8/5) 2^{-j}/f_0^2$. Hence, the full product can be bounded by
\begin{align}
    \prod_{j=r}^k\left(1-\frac{2^{-j}}{f_0^2}\right)^{-1} & \leq 4 \exp \left( \frac{8}{5}\sum_{j=\max{2,r}}^k \frac{2^{-j}}{f_0^2} \right)
     \leq 4 \exp \left( \frac{24}{5}\sum_{j=2}^\infty 2^{-j-1} \right) = 4e^{6/5} < 13.5\ .
\end{align}
Hence,
\begin{align}
d_r
&\leq
\frac{13.5}{A}
\sum_{m=0}^\infty
\left(\frac BA\right)^m
=
\frac{13.5}{A(1-B/A)}
=
\frac{13.5}{A-B}
=
\frac{13.5}{\mu},
\qquad 1\leq r\leq n.
\end{align}
which follows from $B/A < 1$. Then, $d_0 = f_0^{-1}+(1-f_0)/f_0 d_1$. Since, $\mu \leq 1/2$ and $f_0 \geq \sqrt{2/3}$, we again have that $d_0 \leq 13.5/\mu$. Consequently,
\begin{align}
\E[T\vert X_0 = n] = g(n)
=
\sum_{r=0}^n d_r
\leq
\frac{13.5(n+1)}{\mu}.
\label{eq:tolerant-g-bound}
\end{align}

We have now shown that $(Y_t)_{t\in\mathbb N_+}$ has drift exactly $1$ toward zero and increments bounded by $13.5/\mu$. By Kötzing's concentration theorem~\cite[Theorem 2]{Kotzing16},
\begin{align}
    \Pr[T \geq R] \leq \exp\left( \frac{-R\mu^2}{1458} \right), \quad R \geq 2g(n)\ .
\end{align}
Since Eq.~\eqref{eq:tolerant-g-bound} gives $g(n)\leq13.5(n+1)/\mu$, it is enough to take $R\geq2\lceil13.5(n+1)/\mu\rceil$. Therefore $R=R_\mu(n,\delta)$ implies $\Pr[T\geq R]\leq\delta/2$. By stochastic domination, the candidate list contains $\ket{S_\star}$ with probability at least $1-\delta/2$.

Finally, use fresh copies to estimate all $R$ candidate fidelities simultaneously to additive error $\xi/2$, with failure probability
at most $\delta/2$. Global-Clifford classical shadows have bounded
variance for rank-one projectors, so median-of-means estimation uses
$M=O(\xi^{-2}\log(R/\delta))$ copies~\cite{HKP20}.
Writing $\widehat f_j$ for the estimates, return the candidate with
largest estimated fidelity. If the list contains $\ket{S_\star}$
and all estimates are accurate, then
\begin{align}
\bra{\widehat S}\rho\ket{\widehat S}
&\geq\widehat f_{\widehat S}-\frac\xi2
\geq\widehat f_{S_\star}-\frac\xi2
\geq\FS(\rho)-\xi.
\end{align}
A union bound gives the claimed success probability, and the total copy count is $2R+M$.

The accumulated circuits have $O(nR)$ gates, and constructing all
candidate tableaux costs at most $O(n^2R^2)$ operations.
Each candidate--shadow overlap is computable in $O(n^3)$
operations~\cite{AG04}. Thus the total running time is at most
$O(n^2R^2+n^3RM)$, which implies the stated bound because $R\geq n$.
\end{proof}

\section{Tolerant testing without quantum memory}

We now use the tolerant learner to test whether the input state is close
to any stabilizer state. Recall that the testing problem has two
promises. In the YES case, there exists a stabilizer state whose
fidelity with $\rho$ is at least $1-\varepsilon_1$. In the NO case,
every stabilizer state has fidelity at most $1-\varepsilon_2$. 

The tester has three steps. First, we run the candidate-generation stage
of the tolerant learner, which produces a list of stabilizer states that
might be close to $\rho$. On a YES instance, this list contains a
closest stabilizer with high probability. Second, we use global-Clifford
classical shadows~\cite{HKP20} to obtain a coarse estimate of the
fidelity of every state in the list and discard those whose fidelity is
clearly too small. A simple geometric property of stabilizer states
implies that at most six stabilizers can remain after this filtering
step. Finally, we test these remaining stabilizers directly on fresh
copies of $\rho$. Since only constantly many states remain, this final
validation step has the optimal dependence on the gap
$\varepsilon_2-\varepsilon_1$.

\begin{theorem}[Tolerant single-copy testing]
\label{thm:tolerant-testing}
Fix a known $f_0\in(\sqrt{2/3},1]$, let
$0\leq\varepsilon_1\leq1-f_0$ and
$\varepsilon_1<\varepsilon_2<1$, and let $0<\delta\leq1/3$.
There is a polynomial-time adaptive single-copy tester that distinguishes
\begin{equation}
\label{eq:tolerant-promises}
\Fstab(\rho)\geq1-\varepsilon_1
\qquad\text{from}\qquad
\Fstab(\rho)\leq1-\varepsilon_2
\end{equation}
with probability at least $1-\delta$, using
\[
O\!\left(
n+
\frac{\varepsilon_2}{(\varepsilon_2-\varepsilon_1)^2}
\log\frac1\delta
\right)
\]
copies. The input state may be mixed.

For constant $\delta$ and $0<\varepsilon_2\leq1/5$, this scaling is
optimal up to constant factors depending only on $f_0$, even if the
lower bound is restricted to pure inputs; see
Lemma~\ref{lem:testing-witness}. In particular, ordinary single-copy
stabilizer testing has sample complexity
$\Theta(n+1/\varepsilon)$ for $0<\varepsilon\leq1/5$ at constant
confidence.
\end{theorem}

\begin{proof}
Let $\Delta:=\varepsilon_2-\varepsilon_1$. We first construct a short
list of stabilizer states that could have large fidelity with $\rho$.
Run only the candidate-generation stage from the proof of
Theorem~\ref{thm:tolerant} for $R:=R_\mu(n,2\delta/3)$ rounds and
remove duplicate states. On a YES instance, except with probability
at most $\delta/3$, the resulting list contains a closest stabilizer
$\ket{S_\star}$, with
$\bra{S_\star}\rho\ket{S_\star}=\Fstab(\rho)\geq1-\varepsilon_1$.
Since $f_0$ is fixed, $R=O(n+\log(1/\delta))$.

Since one stabilizer is recorded in each round, the list contains at
most $R$ distinct states. Testing all of them separately would be too
costly, so we first use classical shadows to identify the few states
that can have large fidelity with $\rho$. For each stabilizer
$\ket{S_j}$ in the list, let $f_j:=\bra{S_j}\rho\ket{S_j}$. Using
fresh global-Clifford classical shadows~\cite{HKP20}, obtain estimates
$\widehat f_j$ such that $|\widehat f_j-f_j|\leq1/32$ simultaneously
for every state in the list, except with probability at most
$\delta/3$. The same shadow samples are used for all estimates. Since
there are at most $R$ rank-one projectors and the desired accuracy is
constant, this requires only $O(\log(R/\delta))$ additional
copies~\cite{HKP20}.

Keep only those stabilizers for which $\widehat f_j\geq25/32$. On the
simultaneous accuracy event, every surviving state satisfies
$f_j\geq25/32-1/32=3/4$. Conversely, on a YES instance a closest
stabilizer $\ket{S_\star}$ survives, since
$f_\star=\Fstab(\rho)\geq1-\varepsilon_1\geq f_0$ and hence
$\widehat f_\star\geq f_0-1/32>25/32$.

There can be at most six distinct surviving stabilizers. Indeed,
distinct stabilizer states satisfy
$|\langle S_i|S_j\rangle|\leq1/\sqrt2$~\cite{GarciaMarkovCross14}.
For $J$ surviving states, let
$A:=\sum_{j=1}^J\ket{S_j}\bra{S_j}$ and
$G_{ij}:=\langle S_i|S_j\rangle$. The nonzero eigenvalues of $A$ and
the Gram matrix $G$ coincide. Hence, by Gershgorin's theorem,
\[
\|A\|_\infty
=\lambda_{\max}(G)
\leq\max_i\sum_{j=1}^J|G_{ij}|
\leq1+\frac{J-1}{\sqrt2}.
\]
Since every surviving state has fidelity at least $3/4$, we have
$3J/4\leq\tr(\rho A)\leq\|A\|_\infty\leq1+(J-1)/\sqrt2$. Therefore $J\leq6$.

If no state or more than six states survive this filtering step, we
reject. Otherwise, let $\ket{S_1},\ldots,\ket{S_J}$, with
$1\leq J\leq6$, denote the surviving stabilizer states. Thus, on a
YES instance, except with probability at most $2\delta/3$, this short
list contains a closest stabilizer.

We now test these remaining stabilizers directly. Write
$\ket{S_j}=C_j^\dagger\ket{z_j}$. On a fresh copy of $\rho$, apply
$C_j$, measure in the computational basis, and call the trial a
rejection unless the outcome is $z_j$. Its rejection probability is
$q_j:=1-\bra{S_j}\rho\ket{S_j}$.

The two promises now take a particularly simple form. On a YES
instance, some remaining stabilizer satisfies
$q_j\leq\varepsilon_1$. On a NO instance, every stabilizer state, and
hence every state in the list, satisfies $q_j\geq\varepsilon_2$.

For each remaining stabilizer, repeat this measurement independently
$M:=\left\lceil10\varepsilon_2\Delta^{-2}\log(36/\delta)\right\rceil$
times, and let $\widehat q_j$ denote the empirical rejection rate. We
accept if $\widehat q_j\leq(\varepsilon_1+\varepsilon_2)/2$ for at
least one $j$.

An incorrect validation decision requires a deviation of at least
$\Delta/2$. By monotonicity of Bernoulli tails, it is enough to
consider the boundary values $q_j=\varepsilon_1$ and
$q_j=\varepsilon_2$. At either boundary the Bernoulli variance is at
most $\varepsilon_2$, so Bernstein's inequality gives
\begin{align}
\Pr[\text{validation error for }S_j]
&\leq
2\exp\!\left[
-\frac{M(\Delta/2)^2}
{2(\varepsilon_2+\Delta/6)}
\right] \leq
2\exp\!\left[
-\frac{M\Delta^2}{10\varepsilon_2}
\right]
\leq
\frac{\delta}{18},
\end{align}
where we used $\Delta\leq\varepsilon_2$. Since $J\leq6$, a union
bound makes the total validation error at most $\delta/3$.

On a YES instance, candidate generation, the classical-shadow
estimation, and final validation therefore fail with total probability
at most $\delta$. On a NO instance, the algorithm either rejects
earlier or reaches the final stage with at most six stabilizer states,
all satisfying $q_j\geq\varepsilon_2$; hence its probability of
accepting is at most $\delta/3$. Thus the tester succeeds with
probability at least $1-\delta$ in both cases.

It remains to count copies. Candidate generation uses
$O(n+\log(1/\delta))$ copies. The classical-shadow step uses
$O(\log(R/\delta))=O(\log n+\log(1/\delta))$, which is absorbed by
$O(n+\log(1/\delta))$. Since at most six stabilizers are tested in the
final stage, validation uses
$O(\varepsilon_2\Delta^{-2}\log(1/\delta))$ copies. Finally,
$\Delta\leq\varepsilon_2<1$ implies
$\varepsilon_2/\Delta^2\geq1$, so the standalone
$\log(1/\delta)$ term is absorbed as well. Therefore the total number
of copies is $O\!\left(
n+
\frac{\varepsilon_2}{(\varepsilon_2-\varepsilon_1)^2}
\log\frac1\delta
\right),$ as claimed. All stages run in polynomial time.
\end{proof}

To prove optimality, we combine the known linear lower bound with a simple one-qubit family capturing the dependence on the tolerant gap.

\begin{lemma}[Testing lower bound]
\label{lem:testing-witness}
Let $0\leq\varepsilon_1<\varepsilon_2\leq1/5$. Any adaptive single-copy tester that distinguishes
$\Fstab(\rho)\geq1-\varepsilon_1$ from
$\Fstab(\rho)\leq1-\varepsilon_2$
with any fixed success probability greater than $1/2$ requires
\begin{equation}
\Omega\!\left(
n+\frac{\varepsilon_2}{(\varepsilon_2-\varepsilon_1)^2}
\right)
\end{equation}
copies in the worst case, even when the input is promised to be pure.
\end{lemma}

\begin{proof}
For $0\leq\varepsilon\leq1/5$, define
$\ket{\eta_\varepsilon}:=
(\sqrt{1-\varepsilon}\ket{0}
+e^{i\pi/4}\sqrt{\varepsilon}\ket{1})
\otimes\ket{0}^{\otimes(n-1)}$.
The Bloch vector of the first qubit is
$(\sqrt{2\varepsilon(1-\varepsilon)},
\sqrt{2\varepsilon(1-\varepsilon)},1-2\varepsilon)$.
For $\varepsilon\leq1/5$, its largest component is the $Z$ component, so its closest one-qubit stabilizer is $\ket{0}$, with fidelity $1-\varepsilon$. Moreover, no entangled stabilizer can have larger overlap: postselecting the last $n-1$ qubits of any $n$-qubit stabilizer state onto $\ket{0}^{\otimes(n-1)}$ gives either zero or an unnormalized one-qubit stabilizer state. Hence
$\Fstab(\eta_\varepsilon)=1-\varepsilon$.

Set $\Delta:=\varepsilon_2-\varepsilon_1$. A direct calculation gives
\begin{equation}
\label{eq:testing-witness-overlap}
1-\left|
\langle\eta_{\varepsilon_1}\mid\eta_{\varepsilon_2}\rangle
\right|^2
=
\frac{\Delta^2}{
\bigl(
\sqrt{\varepsilon_2(1-\varepsilon_1)}
+\sqrt{\varepsilon_1(1-\varepsilon_2)}
\bigr)^2
}
=
\Theta\!\left(\frac{\Delta^2}{\varepsilon_2}\right),
\end{equation}
where the denominator lies between $(4/5)\varepsilon_2$ and $4\varepsilon_2$.
Thus the two states have squared overlap
$1-\Theta(\Delta^2/\varepsilon_2)$. By the Helstrom bound~\cite{Watrous18}, distinguishing their $N$-copy states with any fixed success probability greater than $1/2$ requires
$N=\Omega(\varepsilon_2/\Delta^2)$, even allowing arbitrary collective measurements.

Independently, Ref.~\cite[Theorem~1.1]{AS26} gives an $\Omega(n)$ lower bound for adaptive single-copy stabilizer testing with pure inputs at constant infidelity, in particular at infidelity $1/5$. This lower bound applies to the present promise problem because exact stabilizer states satisfy the YES promise, while states with stabilizer fidelity at most $4/5$ satisfy the NO promise whenever $\varepsilon_2\leq1/5$.
Combining the two lower bounds gives the claimed lower bound.
\end{proof}
\section{Testing with quantum memory}

We now consider memory-assisted testing. Here, the tester has access to $0\leq k \leq n$ coherent memory qubits. The case of $k=0$ corresponds to single-copy measurements and $k=n$ to two-copy measurements. Previously, an $O((m+1)/\varepsilon)$ tester was known~\cite{AS26}, where $m=n-k$, and our protocol substantially improves the $\varepsilon$ dependence:

\begin{theorem}[Memory-assisted optimal tester]\label{thm:memory}
Let $n\geq1$, $0\leq k\leq n$, and $0<\varepsilon\leq1/5$. Given independent copies of an unknown $n$-qubit state $\rho$, which may be mixed, there is a polynomial-time tester with $k$ qubits of persistent quantum memory that distinguishes $\FS(\rho)=1$ from $\FS(\rho)\leq1-\varepsilon$ with probability at least $5/6$, using $O(m+1/\varepsilon)$ copies, where $m:=n-k$. Every such tester requires $\Omega(m+1/\varepsilon)$ copies, so this scaling is optimal.
\end{theorem}

The tester is composed of two steps. In the first step, a Clifford circuit is learned that unentangles the last $m$ qubits from the first $k$ when the input is a stabilizer state. Then, a two-copy testing algorithm is run on the remaining $k$ qubits. For the first step, an important subroutine will be partial Bell sampling, Algorithm~\ref{alg:partial-bell_samp}. This routine takes $4$ copies of a state and returns $(\alpha,v) \in \F^{2k} \times \F^m$. The following facts about this primitive will be useful.

\begin{lemma}[Corollary 3.3 of Ref.~\cite{AS26}]\label{lem:partial_bell}
Let $\ket{\psi}$ be a stabilizer state with Lagrangian subspace $L$ representing its stabilizer group. The probability of observing $(\alpha,v)$ from partial difference sampling of $\ket{\psi}$ is given by
\begin{align}
\frac{\left|\{w\in\F^m:(\alpha_Xv,\alpha_Zw)\in L\}\right|}{2^n}\ .
\end{align}
\end{lemma}

\begin{lemma}[Probability of partial Bell sampling returning null vectors]\label{lem:bell_samp_null}
Let $D_m:=\{w\in\F^m:(0^k0^m,0^kw)\in L\}$. Then the probability of partial Bell sampling returning $v=0$ is at most $\vert D_m\vert / 2^m$.
\end{lemma}

\begin{proof}
From Lemma~\ref{lem:partial_bell}, we know that the probability of observing $v=0$ is given by
\begin{align}
\frac{\left|\{(\alpha,w)\in\F^{2k}\times\F^m:(\alpha_X0^m,\alpha_Zw)\in L\}\right|}{2^n}\ .
\end{align}
Let $L_k:=\{\alpha\in\F^{2k}:\exists w\in\F^m,\ (\alpha_X0^m,\alpha_Zw)\in L\}$ be the subspace of prefixes. Since the last $m$ $X$-bits are zero, this must be an isotropic set and hence is of size no more than $2^k$. For each $\alpha\in L_k$, the possible suffixes $w$ form one coset of $D_m$, so there are $|D_m|$ such lifts. Hence, the number of stabilizers with no $X$ action on the last $m$ qubits is at most $\vert L_k\vert \vert D_m\vert \leq 2^k \vert D_m \vert$.
\end{proof}

Notice that Lemma~\ref{lem:partial_bell} implies that we can observe $(\alpha,v)$ under partial Bell sampling if and only if there exists a $w \in \F^m$ such that $(\alpha_X v, \alpha_Z w)$ is a stabilizer. When $v \neq 0$, we can employ a strategy similar to the single-copy learning algorithm to first map the stabilizer to one of the form $(\alpha_X e_i, \alpha_Z w')$, where $i\in\{1,\ldots,m\}$ satisfies $v_i=1$ and $j=k+i$ is the global pivot-qubit index. From here, we can remove the action on the first $k$ qubits and, with probability $1/2$, create a $Z$-type stabilizer which acts only on the last $m$ qubits and must act on qubit $j$. That is, 
if $v\neq0$, the update never decreases $\dim D_m$ and increases it with probability at least $1/2$. Once $D_m = \F^m$, the last $m$ qubits must be in some computational basis state and unentangled from the first $k$.

\begin{lemma}[Memory-efficient Clifford disentangling of stabilizer states]\label{lem:unentangle}
Let $\ket\psi$ be a stabilizer state and $m=n-k$. For $0<\delta\leq1/3$, the protocol below uses $k$ qubits of persistent quantum memory and $O(m+\log(1/\delta))$ copies to learn a Clifford $C$ such that, with probability at least $1-\delta$, $C\ket\psi=\ket\phi\otimes\ket z$, where $\ket\phi$ is a $k$-qubit stabilizer state and $\ket z$ is an $m$-qubit computational-basis state. If $m=0$, take $C=I$ and use no compression rounds.
\end{lemma}

\begin{proof}
Let $D_m$ be the subspace of labels $w\in\F^m$ for which $\id_k\otimes Z(w)$ is a stabilizer up to sign. These are the diagonal stabilizer labels acting only on the last $m$ qubits. When $\log_2 \vert D_m \vert = m$, the last $m$ qubits must be in the state $\ket{z}$ for some $z\in \F^m$.

Say that $(\alpha,v)$ is obtained from partial Bell sampling where $v\neq 0$. The next step of the algorithm is to map the stabilizer $(\alpha_X v, \alpha_Z w)$, which exists for some $w$ via Lemma~\ref{lem:partial_bell}, to one of the form $(00,0w')$. As we have previously seen,
\begin{align}
    (\alpha_Xv,\alpha_Zw)\overset{F_v}{\longmapsto}(\alpha_Xe_i,\alpha_Zw')\ ,
\end{align}
where $i$ is the first index with $v_i=1$, $j=k+i$, $F_v$ acts on the last $m$ qubits, and $w'\in\F^m$. Choose a Hermitian Pauli representative $P_\alpha$ on the first $k$ qubits, so $P_\alpha^2=I$. Define the controlled-Pauli Clifford
\begin{align}
    W_{\alpha,j} & = \id_k \otimes \ket{0}\bra{0}_j + P_\alpha \otimes \ket{1}\bra{1}_j \ ,
\end{align}
which we define as acting trivially on all other qubits. That is, $W_{\alpha,j}$ applies the Pauli string $P_\alpha$ to the first $k$ qubits if qubit $j$ is in the state $\ket{1}$ and does nothing otherwise. One can verify that
\begin{align}
    W_{\alpha,j}(P_\alpha \otimes X_j) W_{\alpha,j} & = \id_k \otimes X_j\ , \\
    W_{\alpha,j}(P_\alpha \otimes Y_j) W_{\alpha,j} & = \id_k \otimes Y_j\ ,
\end{align}
where we have suppressed identities on the remaining $m-1$ qubits for ease of notation. Hence, $W_{\alpha,j}$ takes the stabilizer $(\alpha_X e_i, \alpha_Z w')$ to $(0^ke_i,0^kw')$. As in the case of the learning algorithm, one choice of $b\in\{0,1\}$ further reduces this to a $Z$-type stabilizer that acts nontrivially only on the last $m$ qubits and must include a $Z_j$.

It remains to argue that this process cannot decrease the size of $D_m$, only increase it. As we have seen before, $F_v$ does not change the size of $D_m$ at all. Since all bitstrings in $D_m$ after applying $F_v$ have zero in local coordinate $i$, $W_{\alpha,j}$ leaves $D_m$ unchanged. Hence, when $v\neq0$, each round preserves $\dim D_m$ or increases it, with an increase of at least one occurring with probability at least $1/2$. By Lemma~\ref{lem:bell_samp_null}, the probability of a single step increasing the dimension of $D_m$ by at least $1$ is at least $\frac{1}{2}(2^m-\vert D_m \vert)/2^m$. Similarly to the prior analysis of the learning algorithm, $2m+O(\sqrt{m\log(1/\delta)}+\log(1/\delta))$ steps then suffice to succeed with probability at least $1-\delta$.
\end{proof}

\begin{algorithm}[t]
\caption{Partial Bell difference sampling using $k$ qubits of memory}
\label{alg:partial-bell_samp}
\KwInput{Four copies of $\ket\psi$.}
\KwOutput{$(\alpha,v)$ with $\alpha\in\F^{2k}$ and $v\in\F^m$.}
\For{$t=0$ to 1}{
Store the first $k$ qubits of one copy of $\ket\psi$;
Measure the remaining $m$ qubits in the computational basis, obtaining
$z_{0}^t\in\F^m$;
Load another copy of $\ket\psi$;
Bell-sample its first $k$ qubits against the $k$ qubits in memory,
obtaining $\alpha^t\in\F^{2k}$;
Measure the remaining $m$ qubits in the computational basis, obtaining
$z_{1}^t\in\F^m$;
}
\Return{$(\alpha^0\oplus \alpha^1,\ z_0^0 \oplus z_1^0 \oplus z_0^1 \oplus z_1^1)$};
\end{algorithm}

\begin{algorithm}[t]
\caption{Stabilizer testing with $k$ qubits of memory}
\label{alg:k_memory_test}
\KwInput{$4T_1+1+6T_2$ independent copies of $\rho$, where $m=n-k$, $T_1=2m+O(\sqrt{m\log(1/\delta)}+\log(1/\delta))$ is chosen for compression error at most $\delta$, and $T_2=\lceil8\log(1/\delta)/(3\eps)\rceil$. Take $T_1=0$ if $m=0$.}
$C\leftarrow I$\;
\For{$t = 1$ to $T_1$}{
$(\alpha, v)\leftarrow$ partial Bell difference sampling $C \ket{\psi}$\;
If $v\neq0$, let $i$ be the first index with $v_i=1$, set $j=k+i$, sample an independent uniform bit $b$, and update
\begin{align}
C \leftarrow H_jS_j^bW_{\alpha, j} F_v C\ .
\end{align}
If $v=0$, leave $C$ unchanged.
}
Measure the last $m$ bits of one copy of $C\ket{\psi}$ to obtain one sample $z \in \F^m$\;
\For{$t=1$ to $T_2$}{
Before using each copy, measure the last $m$ qubits of $C\ket{\psi}$ in the computational basis. If a string other than $z$ is ever observed, $\mathsf{Reject}$\;
Otherwise, run the six-copy routine of Ref.~\cite[Theorem~3.7]{iyer2024tolerant} on the first $k$ qubits and $\mathsf{Reject}$ if its output is $-1$ (for $k=0$, this routine always outputs $+1$)\;
}
$\mathsf{Accept}$\;
\end{algorithm}

We can now present the complete tester, which appears in Algorithm~\ref{alg:k_memory_test}. The first step is to learn a Clifford that would unentangle the state if it was indeed a stabilizer state. The second is to run the $\{+1,-1\}$-valued routine of Ref.~\cite[Theorem~3.7]{iyer2024tolerant}, accepting on $+1$. It uses six copies, but only two at a time and no ancillary qubits; on the residual $k$-qubit subsystem it therefore fits the $k$-qubit persistent-memory budget. To conclude, we prove that this tester successfully distinguishes between stabilizer and non-stabilizer states.

\begin{proof}[Proof of Theorem~\ref{thm:memory}]
Fix $0<\delta\leq1/3$; the theorem follows by taking $\delta=1/6$. Lemma~\ref{lem:unentangle} ensures that we can learn a circuit that would unentangle the last $m$ qubits for a stabilizer state with probability at least $1-\delta$ using $O(m+\log 1/\delta)$ samples.
If the state is a stabilizer state, the test of Ref.~\cite{iyer2024tolerant} passes with probability $1$, and hence a stabilizer state is accepted with probability at least $1-\delta$.

Now suppose $\FS(\rho)\leq1-\eps$, and condition on the complete training transcript and the chosen string $z$. In this case, $C$ need not disentangle the last $m$ qubits. Let $p$ be the probability of observing the fixed computational-basis state $\ket{z}$ when measuring the last $m$ qubits of $C\rho C^\dagger$. If $p=0$, every verification block rejects. Otherwise, let $\tau=p^{-1}\bra z C\rho C^\dagger\ket z$ be the normalized conditional state on the first $k$ qubits. For every $k$-qubit stabilizer $\ket S$, the state $C^\dagger(\ket S\otimes\ket z)$ is a stabilizer, so $p\bra S\tau\ket S\leq\FS(\rho)$. Maximizing over $S$ gives $p\FS(\tau)\leq\FS(\rho)$. Theorem A.3 of Ref.~\cite{iyer2024tolerant} implies that the probability that the tester falsely passes, when $z$ was measured on all copies, is at most $(5+3\FS(\tau))/8$. Hence, the probability that one run of the tester, including the measurements of $z$, passes, is at most
\begin{align}
    p^6\frac{5+3\FS(\tau)}{8}
    &\leq p^5\frac{5p+3\FS(\rho)}{8}\\
    &\leq 1-\frac{3\eps}{8}\ .
    \nonumber
\end{align}
The verification blocks use new copies and are independent conditional on the training transcript. Repeating $T_2=\lceil8\log(1/\delta)/(3\eps)\rceil$ times bounds the probability that all blocks accept by $(1-3\eps/8)^{T_2}\leq\delta$. Together with completeness, this gives the stated guarantee. The total copy cost is $O(m+\log(1/\delta)+\eps^{-1}\log(1/\delta))$.

For optimality at constant confidence, Ref.~\cite[Theorem~1.1 and Appendix~B]{AS26} gives the $\Omega(n-k)$ lower bound with pure inputs. Lemma~\ref{lem:testing-witness}, with $a=0$ and $b=\eps\leq1/5$, gives $\Omega(1/\eps)$ copies even for possibly collective measurements. Taking the maximum of these two lower bounds gives $\Omega(n-k+1/\eps)$.
\end{proof}

\section{Learning states with low stabilizer nullity}\label{sec:low-nullity}

The adaptive mechanism is not restricted to exact stabilizer states.
We now show that it also applies when only a large Pauli stabilizer subspace remains. 
For a pure state $\ket{\psi}$, recall that $L_\psi\leq\F^{2n}$ denotes its unsigned stabilizer space and define its \emph{stabilizer nullity} by $\nu(\psi):=n-\dim L_\psi$.
In the extreme case, pure states with $\nu(\psi)=0$ are stabilizer states. 
When $\nu(\psi)\leq r$, the difference samples used by our learning algorithm still reveal a genuine stabilizer direction with probability $2^{-O(r)}$, while the remaining difference samples do not destroy the diagonal stabilizers already exposed.
This key observation enables the following theorem.

\begin{theorem}[Learning states of low stabilizer nullity]
\label{thm:low-nullity}
Let $n\geq1$, $0\leq r\leq n$ be an integer, and let $\ket\psi$ be an unknown pure $n$-qubit state with $\nu(\psi)\leq r$.
For $0<\eps\leq1$ and $0<\delta\leq1/3$, there is an adaptive single-copy protocol that outputs a classical description of a state $\widehat\rho$ satisfying $\tfrac12\|\widehat\rho-\ket{\psi}\bra{\psi}\|_1\leq\eps$ with probability at least $1-\delta$, using
\begin{equation}
\label{eq:low-nullity-complexity}
N=O\left(2^r\left(n+\log\frac1\delta\right)+\frac{2^r}{\eps^2}\left(r+\log\frac1\delta\right)\right)
\end{equation}
copies of $\ket{\psi}$. 
The classical running time is polynomial in $n$, $2^r$, $1/\eps$, and $\log(1/\delta)$.
\end{theorem}

Before we prove Theorem~\ref{thm:low-nullity}, we remark that states with $\nu(\psi)\leq r$ also include those prepared by circuits doped with at most $\lfloor r/2\rfloor$ single-qubit non-Clifford gates or $r$ $T$ gates in the Clifford+$T$ setting. A gate on at most $q$ qubits can reduce stabilizer dimension by at most $2q$, because the subspace acting trivially on those qubits is preserved; a $T$ gate preserves the subspace with no $X$ component on its qubit and reduces the dimension by at most one.
Thus, Theorem~\ref{thm:low-nullity} provides a learner whose dependence on $n$ is linear for fixed $t$, accuracy, and confidence.

\begin{proof}
We run the same adaptive rounds as in Theorem~\ref{thm:explicit-exact}
(Algorithm~\ref{alg:stab_learn}).
Fix an iteration $1\leq t\leq T$, and let $\ket{\varphi}=C\ket{\psi}$ denote the state after the Clifford transformations chosen in the previous rounds.
Let $r_0:=\nu(\psi)=\nu(\varphi)\leq r$, where we have used the fact that Clifford transformations preserve stabilizer
nullity.
We keep the notation in the proof of
Theorem~\ref{thm:explicit-exact}, and recall that, given a pure state $\ket{\varphi}$, $L_\varphi\leq\F^{2n}$ is the unsigned stabilizer space, $D(\varphi)=\{s_z:(0,s_z)\in L_\varphi\}$ is the space of $Z$-type stabilizers, and $\ell=\log_2(|D(\varphi)|)=\dim D(\varphi)$ is the dimension of $D(\varphi)$.

We first characterize which strings in $D(\varphi)^\perp$ can occur as the $X$-part of a stabilizer. 
Consider the linear map from $L_\varphi$ to $\F^n$ given by $(s_x,s_z)\mapsto s_x$ with kernel $\{0\}\times D(\varphi)$. 
Since $\dim L_\varphi=n-r_0$, the set of $u\in\F^n$ for which there exists some $v\in\F^n$ such that $(u,v)\in L_\varphi$ is a linear subspace of dimension $n-r_0-\ell$, and every such $u$ belongs to $D(\varphi)^\perp$. 
Moreover, for every $s\in D(\varphi)$, the two stabilizers $(u,v)$ and $(0,s)$ commute, and hence $\langle u,s\rangle=0$.
Since $\dim D(\varphi)^\perp=n-\ell$, the set of strings in
$D(\varphi)^\perp$ which occur as the $X$-part of a stabilizer has codimension $r_0$. Denote this subspace by $V:=\pi_X(L_\varphi)$.

We next consider the distribution obtained from computational difference sampling. 
As in the proof of Theorem~\ref{thm:explicit-exact}, the computational-basis distribution of $\ket{\varphi}$ is supported on an affine coset of $D(\varphi)^\perp$. 
In addition, whenever $(u,v)\in L_\varphi$, the corresponding Pauli stabilizer maps $\ket{x}$ to $\ket{x+u}$, up to a phase. 
It follows that the computational-basis probabilities are invariant under translation by every $u$ which occurs as the $X$-part of a stabilizer.
Therefore, the computational-basis distribution is uniform within each coset of this $(n-r_0-\ell)$-dimensional subspace. 
There are $2^{r_0}$ such cosets inside the relevant affine coset of $D(\varphi)^\perp$. 
Let $p_a$ denote the probability mass of the $a$-th coset. 
If $x,y$ are two independent computational-basis samples and $u=x+y$, then $u$ is the $X$-part of a stabilizer whenever $x$ and $y$ lie in the same coset. 
Hence, we have
\begin{equation}
\Pr\left[\exists v\in\F^n:\ (u,v)\in L_\varphi\right]=\sum_a p_a^2\geq 2^{-r_0}.
\end{equation}
Here, the final inequality follows from Cauchy--Schwarz.
Conditioned on $x$ and $y$ lying in the same coset, their difference is uniformly distributed over the corresponding $(n-r_0-\ell)$-dimensional linear subspace. 
Hence, as long as $\ell<n-r_0$,
\begin{equation}
\label{eq:low-nullity-nonzero}
\Pr\left[u\neq0\text{ and }\exists v\in\F^n:\ (u,v)\in L_\varphi\right]\geq2^{-r_0}\left(1-2^{-(n-r_0-\ell)}\right).
\end{equation}
Importantly, even when $u$ is not the $X$-part of a stabilizer, 
the update in Algorithm~\ref{alg:stab_learn} cannot decrease $\ell$. 
Since both computational-basis outcomes belong to the same affine coset of $D(\varphi)^\perp$, every difference sample satisfies $u=x+y\in D(\varphi)^\perp$.
Suppose that $u\neq0$, and let $i$ be the first index such that $u_i=1$.
We apply the same CNOT circuit $F_u$ as in the proof of
Theorem~\ref{thm:explicit-exact}. Recall that its action on a $Z$-type stabilizer is
\begin{equation}
(0,s)\stackrel{F_u}{\longmapsto}\left(0,s+e_i\langle u+e_i,s\rangle\right).
\end{equation}
Since $u\in D(\varphi)^\perp$, all $s\in D(\varphi)$ satisfy $\langle u,s\rangle=0$. 
As in the proof of Theorem~\ref{thm:explicit-exact}, it follows that every pre-existing $Z$-type stabilizer acts trivially on qubit $i$ after applying $F_u$.
Hence, the subsequent application of either $H_i$, which is the $H$ gate on the $i$-th qubit, or $H_iS_i$, which is the $S$ gate followed by the $H$ gate on the $i$-th qubit, leaves all existing $Z$-type stabilizers unchanged. 
Thus, $\ell$ can never decrease.

Now suppose that $u\neq0$ and there exists some $v\in\F^n$ such that $(u,v)\in L_\varphi$.
Then $F_u$ maps this stabilizer to one of the form $(u,v)\stackrel{F_u}{\longmapsto}(e_i,v')$ for some $v'\in\F^n$. 
As in Theorem~\ref{thm:explicit-exact}, one of the two choices $H_i$ and $H_iS_i$ maps this stabilizer to a new $Z$-type stabilizer that acts nontrivially on qubit $i$.
Since all existing $Z$-type stabilizers act trivially on qubit $i$, this stabilizer is independent of $D(\varphi)$, and hence $\ell$ increases by one. 
The correct correction is chosen with probability $1/2$.
Combining this with Eq.~\eqref{eq:low-nullity-nonzero}, whenever
$\ell<n-r_0$, a round increases $\ell$ with probability at least
\begin{equation}
\label{eq:low-nullity-progress}
p_\ell\geq\frac{2^{-r_0}}{2}\left(1-2^{-(n-r_0-\ell)}\right)\geq 2^{-r_0-2}\geq2^{-r-2}.
\end{equation}

We now bound the number of rounds required to reach $\ell=n-r_0$. 
When $n-r_0-\ell=j$, Eq.~\eqref{eq:low-nullity-progress} gives a transition probability at least $q_j:=2^{-r_0-1}(1-2^{-j})$.
As before, an unsuccessful round does not decrease $\ell$. Hence, if $G_j\sim\operatorname{Geom}(q_j)$ are independent geometric random variables, the stopping time is stochastically dominated by $T_{\mathrm{stop}}\preceq\sum_{j=1}^{n-r_0}G_j$.
Its expectation is therefore bounded by
\begin{align}
\mathbb E[T_{\mathrm{stop}}]\leq2^{r_0+1}\sum_{j=1}^{n-r_0}\frac{1}{1-2^{-j}}=O(2^{r_0}n)=O(2^rn).
\end{align}
Since $q_j\geq2^{-r_0-2}$ for all $j\geq1$, the same upper-tail bound for sums of independent geometric random variables used in Theorem~\ref{thm:explicit-exact} gives
\begin{equation}
T_{\mathrm{stop}}=O\left(2^r\left(n+\log\frac1\delta\right)\right)
\end{equation}
with probability at least $1-\delta/4$. Each round uses two copies.
Let $C$ denote the Clifford circuit obtained at the end of this procedure and write $\ket{\varphi}=C\ket{\psi}$.
On the event above, we have $\ell=n-r_0$.
At this point, $D(\varphi)$ has the same dimension as $L_\varphi$, and hence all Pauli stabilizers of $\ket{\varphi}$ are $Z$-type, and $\dim D(\varphi)^\perp=n-\ell=r_0$.
Since the computational-basis distribution of $\ket{\varphi}$ is supported on an affine coset of $D(\varphi)^\perp$, the state $\ket{\varphi}$ is supported on an affine subspace of dimension at most $r_0\leq r$.

It remains to learn the state inside this low-dimensional affine subspace.
Take $M+1$ additional computational-basis samples $x_0,x_1,\ldots,x_M$ from $\ket{\varphi}$ and let
\begin{equation}
\label{eq:low-nullity-affine-span}
A=x_0+\operatorname{span}\{x_j+x_0:1\leq j\leq M\}.
\end{equation}
Suppose that at some point $A$ contains probability mass less than $1-\eta$. 
Then the next sample lies outside $A$, and hence increases its affine dimension, with probability larger than $\eta$. 
Since the affine dimension can increase at most $r$ times, the same geometric-variable argument shows that $M=O\left(\frac{r+\log(1/\delta)}{\eta}\right)$ samples suffice to ensure, with probability at least $1-\delta/4$, that $\Pr[x\in A]\geq1-\eta$.

We take $\eta=\varepsilon^2/16$. 
Let $\Pi_A=\sum_{x\in A}\ket{x}\bra{x}$ and define $p=\bra{\varphi}\Pi_A\ket{\varphi}$.
On the event $\Pr[x\in A]\geq1-\eta$, i.e., $p\geq1-\eta$, the normalized projection $\ket{\widetilde{\varphi}}=\frac{\Pi_A\ket{\varphi}}{\sqrt p}$ satisfies
\begin{equation}
\label{eq:low-nullity-projection-error}
\frac12\left\|\ket{\varphi}\!\bra{\varphi}-\ket{\widetilde{\varphi}}\!\bra{\widetilde{\varphi}}\right\|_1=\sqrt{1-p}\leq\frac{\varepsilon}{4}.
\end{equation}
Let $r'=\dim A$. If $r'>r$, the algorithm aborts and outputs a fixed state, for example $(\ket0\bra0)^{\otimes n}$; this occurs only on the compression-failure event. Otherwise, $r'\leq r$. 
Since $A$ is known classically, a Clifford circuit
$U$ can be efficiently constructed such that $UA=\{0\}^{n-r'}\times\F^{r'}$.
It follows that
\begin{equation}
\label{eq:low-nullity-residual-state}
U\ket{\widetilde{\varphi}}=\ket{0}^{\otimes(n-r')}\otimes\ket{\chi}
\end{equation}
for some $r'$-qubit pure state $\ket{\chi}$.

We can now learn the remaining $r'$-qubit state using projected least-squares tomography with random Clifford measurements~\cite{HKP20,GKKT20}.
On each new copy of $\ket{\psi}$, we first apply $C$ and then $U$, measure the first $n-r'$ qubits in the computational basis, and keep the remaining system only when the first $n-r'$ outcomes are all zero.
Conditioned on acceptance, the remaining $r'$ qubits are in the state $\ket{\chi}$, while the acceptance probability is $p\geq1-\eta$.
Conditioned on the complete compression and affine-span recovery transcripts, the accepted residual systems are independent identical copies of $\ket{\chi}$.
Let $d=2^{r'}\leq2^r$. On each accepted copy, we apply an independently and uniformly random $r'$-qubit Clifford and measure in the computational basis.
These measurements are equivalent to sampling from the stabilizer-state projective $2$-design.
By the trace-norm guarantee for projected least-squares tomography with projective $2$-design measurements~\cite[Theorem~1]{GKKT20}, using $O(d\varepsilon^{-2}\log(d/\delta))$ accepted copies produces, with probability at least $1-\delta/4$, a density matrix $\widehat\rho_{\rm log}$ satisfying
$\frac12\|\widehat\rho_{\rm log}-\ket{\chi}\bra{\chi}\|_1\leq\varepsilon/4$.
Here, when $r'=0$, the logical system is trivial and $\widehat\rho_{\rm log}$ is understood to be its unique state.
The estimator can be computed in time polynomial in the number of samples and in $d$~\cite{GKKT20}.
Since $d\leq2^r$, the logical tomography stage therefore uses
\begin{equation}
\label{eq:low-nullity-tomography}
O\left(\frac{d}{\varepsilon^2}\log\frac{d}{\delta}\right)
\leq
O\left(\frac{2^r}{\varepsilon^2}\left(r+\log\frac1\delta\right)\right)
\end{equation}
accepted copies. Finally, since $\eta=\varepsilon^2/16$ and $\varepsilon\leq1$, on the preceding success events the acceptance probability satisfies $p\geq15/16$.
A standard Chernoff bound (see also Ref.~\cite[proof of Theorem~6.4]{GIKL25} for a similar postselection argument) therefore shows that a constant-factor larger number of input copies suffices to obtain the required number of accepted copies except with probability at most $\delta/4$.
Hence the total copy complexity of the tomography stage has the same asymptotic scaling as Eq.~\eqref{eq:low-nullity-tomography}.

The classical output consists of $(C,U,r',\widehat\rho_{\rm log})$, representing
\begin{equation}
 \widehat\rho=C^\dagger U^\dagger\bigl((\ket0\bra0)^{\otimes(n-r')}\otimes\widehat\rho_{\rm log}\bigr)UC.
\end{equation}
The triangle inequality combines the projection error and logical tomography error. A union bound over compression, affine-span recovery, logical tomography, and obtaining enough accepted copies gives an estimate $\widehat\rho$ satisfying $\tfrac12\left\|\widehat{\rho}-\ket{\psi}\!\bra{\psi}\right\|_1\leq \varepsilon$ with probability at least $1-\delta$, using
\begin{equation}
N=O\left(2^r\left(n+\log\frac1\delta\right)+\frac{2^r}{\varepsilon^2}\left(r+\log\frac1\delta\right)\right)
\end{equation}
copies. 
This completes the proof.
\end{proof}

\section{Conclusions and open questions}
In this work, we have comprehensively revisited the problem of learning and testing stabilizer states with optimal sample complexity. We have shown that arbitrary $n$-qubit stabilizer states can be learned exactly and in polynomial time from $\Theta(n)$ copies using adaptive single-copy measurements. Thus, classical adaptivity closes the quadratic gap between non-adaptive single-copy learning~\cite{ABDY23} and two-copy Bell sampling~\cite{Mon17}, without requiring coherent access to multiple copies. At the heart of our approach is an adaptive subroutine that uses measurement outcomes to choose subsequent measurement bases, progressively simplifying the unknown state while preserving the diagonal stabilizer structure already exposed. We have extended this approach to tolerant stabilizer learning and testing in the stated parameter regimes, established an optimal memory--sample tradeoff for stabilizer testing, and obtained efficient single-copy learners for states of low stabilizer nullity.

These results leave several interesting questions about 
the power and limitations of adaptivity.

\begin{itemize}

\item \emph{Additive dependence on stabilizer nullity.}
At constant confidence, Theorem~\ref{thm:low-nullity} learns states of stabilizer nullity at most $r$ using $n2^{O(r)}+2^{O(r)}\eps^{-2}$ single-copy measurements. The factor $2^{O(r)}$ multiplying the linear dependence on $n$ arises because each stabilizer-exposure step waits for two computational-basis outcomes to collide in the same class of the quotient $D^\perp/V$. Is this repeated collision cost intrinsic? In particular, can the multiplicative $n2^{O(r)}$ dependence be replaced by an additive $n+2^{O(r)}$ bound at constant accuracy and confidence? For states of $T$-count at most $t$, such an additive dependence is achievable with two-copy measurements: Bell-sampling-based methods give an $O(n+2^t)$-copy learner at constant accuracy and confidence~\cite{HG24,GIKL25}. On the other hand, one should not expect an information-theoretic $n+2^{\Omega(t)}$ lower bound for the subclass prepared with $t$ single-qubit non-Clifford gates: general bounded-gate-complexity tomography yields a computationally inefficient learner with sample complexity only $\widetilde O(n^2(t+1)/\varepsilon^2)$~\cite{ZLKQHC24}. Thus the interesting question is whether the exponential dependence on $r$ (or $t$) is intrinsic to \emph{computationally efficient single-copy} learning, and, within this model, whether the cost of recovering the extensive stabilizer structure can be separated additively from that of learning the residual non-stabilizer degrees of freedom.

\item \emph{The optimal leading constant.}
For fixed failure probability $0<\delta<1/8$, our algorithm uses $4n+O(\sqrt n)$ copies. In contrast, arbitrary collective measurements achieve $n+O(1)$ copies, with a matching leading-order lower bound~\cite{chang2026sampleoptimallearningstabilizerstates}. Can adaptive single-copy measurements also achieve $n+o(n)$ copies, or does the single-copy restriction impose a strictly larger leading constant? Resolving this question would determine whether adaptivity replaces coherent access to multiple copies only at the level of asymptotic scaling, or also at the level of the optimal leading-order cost.

\item \emph{Fully agnostic stabilizer learning.}
Our tolerant learner assumes a fixed lower bound $f_0>\sqrt{2/3}$ on the stabilizer fidelity. Can this restriction be removed? The goal is to output a stabilizer state $\ket{\widehat S}$ satisfying $\langle\widehat S|\rho|\widehat S\rangle\geq F_{\mathrm{stab}}(\rho)-\xi$ for an arbitrary input state $\rho$. With two-copy measurements, stabilizer bootstrapping achieves this guarantee with runtime
$\operatorname{poly}(n,1/\xi)(1/f)^{O(\log(1/f))}$ when the optimal stabilizer fidelity is at least $f$~\cite{CGYZ25}. Can comparable guarantees be obtained using only adaptive single-copy measurements? A preliminary stabilizer test does not resolve this question: even when the input is far from every stabilizer state, an agnostic learner must still find a competitive stabilizer approximation. The central challenge is therefore to extract useful stabilizer structure without relying on a high-fidelity reference state.

\item \emph{Fermionic analogues.}
Fermionic Gaussian states and their weakly non-Gaussian extensions provide a natural setting in which to investigate related advantages of adaptivity~\cite{pzx6-nkfb,2402.18665}. Can adaptively chosen Gaussian transformations and measurement bases reduce the sample complexity relative to non-adaptive single-copy protocols? For weakly non-Gaussian states, can such updates progressively identify the Gaussian structure while confining the remaining learning problem to a small number of non-Gaussian modes?

\item \emph{When does adaptivity provably help?}
More broadly, can one identify structural conditions under which adaptivity provably reduces the sample complexity of quantum learning tasks? In our algorithm, measurement outcomes do more than provide data about the unknown state: they identify transformations that preserve the structure already exposed while making further progress with constant probability. It would be interesting to understand whether this mechanism---measurement-guided simplification with monotone progress---extends to other structured families of quantum states, and more generally when classical feedback can reproduce information that would otherwise require coherent measurements across multiple copies.

\end{itemize}
\subsection{Acknowledgments}
We thank Srinivasan Arunachalam, Sitan Chen, Arkopal Dutt, Sabee Grewal and Jarrod McClean for discussions about this problem.
L.B., J.E., A.A.M., and L.S.~acknowledge support from the BMFTR (MUNIQC-Atoms, QuSol, Pasquops, Hybrid++, FermiQP), the Munich Quantum Valley, the Quantum Flagship (PasQuanS2, MILLENION), QuantERA (SDPCode), the Clusters of Excellence MATH+ and ML4Q, the DFG (CRC 183 and 
SPP 2514), the Einstein Foundation (Einstein Research Unit on Quantum Devices), Berlin Quantum, and the European Research Council (DebuQC). A.A.M. further acknowledges support from a 2025 Google PhD Fellowship.  W.G.
acknowledges support from NSF Grant
CCF-2430375.

\subsection{AI statement}
The human authors identified that the single-copy lower bound of $\Omega(n^2)$ for learning stabilizer states in Ref.~\cite{ABDY23} was valid only for non-adaptive protocols, and wondered about extending it to adaptive protocols. When GPT-6 came out, we asked GPT-6 whether it could prove an adaptive lower bound of $\Omega(n^2)$ for single-copy stabilizer learning. Instead, it proposed the $\Theta(n)$ adaptive learning algorithm that, after verification and further development by the authors, forms the main result of this work. The human authors subsequently identified stabilizer testing with and without quantum memory, as well as learning beyond stabilizer states, as natural directions to investigate, and prompted GPT-6 Astra to explore these applications. All proofs, arguments, and statements included in the paper were subsequently checked, digested, refined, and rewritten by the human authors.

The main theorems shown here have been formally
verified in Lean~4 using mathlib. OpenAI's Codex assisted with developing,
reviewing, and documenting the formalization.\footnote{The
formalization and reproduction instructions are available at
\url{https://github.com/AntMele/single-copy-stabilizer-learning-lean}.}

\bibliographystyle{unsrtnat}

\appendix

\section{Extension to prime-dimensional qudits}
\label{sec:prime-dimensional-qudits}

The exact learning algorithm extends naturally to qudits of prime local dimension. We briefly introduce the corresponding stabilizer formalism and then explain the modifications to the qubit protocol. In order to proceed, let $p$ be prime, let $\omega=e^{2\pi i/p}$, and let all vector spaces and inner products in this section be over $\mathbb F_p$. On a single qudit, define the Weyl operators by $X(a)\ket{x}=\ket{x+a}$ and $Z(b)\ket{x}=\omega^{bx}\ket{x}$ for $a,b,x\in\mathbb F_p$. For $a,b\in\mathbb F_p^n$, write $P(a,b)=X(a)Z(b)=\bigotimes_{j=1}^n X(a_j)Z(b_j)$. Up to global phases, Pauli operators are therefore labelled by vectors $(a,b)\in\mathbb F_p^{2n}$. Their commutation relations are determined by the symplectic form $[(a,b),(a',b')]=a\cdot b'-b\cdot a'$, so two Paulis commute if and only if this quantity vanishes.

A pure $n$-qudit stabilizer state $\ket{\varphi}$ is specified, up to the corresponding eigenvalue phases, by an $n$-dimensional isotropic subspace $L_\varphi\leq\mathbb F_p^{2n}$. As in the qubit case, define its space of diagonal stabilizer labels by $D(\varphi):=\{b\in\mathbb F_p^n:(0,b)\in L_\varphi\}$ and let $\ell(\varphi):=\dim D(\varphi)$. The computational-basis distribution of $\ket{\varphi}$ is uniform on an affine coset of $D(\varphi)^\perp$, and the projection of the stabilizer space onto its $X$-coordinates satisfies $\pi_X(L_\varphi)=D(\varphi)^\perp$. Hence, if $x,y$ are independent computational-basis measurement outcomes, then their difference $u:=x-y$ is uniform on $D(\varphi)^\perp$, and for every such $u$ there exists some $w\in\mathbb F_p^n$ with $(u,w)\in L_\varphi$.

The adaptive algorithm is the direct analogue of Algorithm~\ref{alg:stab_learn}. Starting with $C=I$, measure two copies of $C\ket{\psi}$ in the computational basis and form $u=x-y\in\mathbb F_p^n$. If $u=0$, no update is made. Otherwise, choose a pivot $i$ with $u_i\neq0$. Generalized controlled-SUM gates, followed by a multiplication Clifford on the pivot, give a Clifford $F_u$ mapping $(u,w)$ to $(e_i,w')$ for some $w'\in\mathbb F_p^n$. Concretely, one may use $\mathrm{SUM}_{i\rightarrow j}(\lambda)\ket{s,t}=\ket{s,t+\lambda s}$ with $\lambda=-u_j/u_i$ to eliminate the non-pivot $X$-components, and then apply the multiplication gate $M_{u_i^{-1}}\ket{s}=\ket{u_i^{-1}s}$ on the pivot.

After this transformation, the local Pauli on the pivot is, up to phase, $XZ^c$ for some unknown $c\in\mathbb F_p$. For each $a\in\mathbb F_p$, fix a one-qudit Clifford $H_a$ satisfying $H_aXZ^aH_a^\dagger\propto Z$. The algorithm chooses $a\in\mathbb F_p$ uniformly at random and applies $H_a$ on the pivot. Thus the correct local correction is selected with probability $1/p$.

Exactly as in the qubit proof, this update never destroys previously exposed diagonal stabilizers. Indeed, since $u\in D^\perp$, after the pivoting circuit every old diagonal stabilizer has zero component on the pivot. The final one-qudit Clifford therefore leaves all of them diagonal. For the unique value of $a$ matching the unknown slope $c$, the stabilizer $(e_i,w')$ is converted into a new diagonal stabilizer acting nontrivially on the pivot, and is consequently linearly independent of the previous diagonal stabilizers. Hence every round preserves or increases $\ell$, and conditioned on $u\neq0$ an increase of at least one occurs with probability at least $1/p$.

\begin{theorem}[Adaptive single-copy learning in prime dimension]
\label{thm:prime-dimensional-learning}
Let $p$ be prime, $n\geq1$, and $0<\delta\leq1/3$. An arbitrary pure stabilizer state of $n$ qudits of local dimension $p$ can be learned exactly, with probability at least $1-\delta$, using
\begin{equation}
T= pn+O_p\left(\sqrt{n\log(1/\delta)}+\log(1/\delta)
\right)
\label{eq:qudit-round-complexity}
\end{equation}
adaptive rounds and hence $N=2T+1$ copies. For every fixed prime $p$, the sample complexity is $\Theta_p(n+\log(1/\delta))$. The algorithm uses only single-copy Clifford operations and computational-basis measurements and is classically efficient.
\end{theorem}

\begin{proof}
Let $r:=n-\ell$ denote the number of diagonal stabilizer generators still missing. Since $u$ is uniform on the $r$-dimensional space $D^\perp$, one has $\Pr[u\neq0]=1-p^{-r}$. Conditioned on $u\neq0$, the correct local correction is selected with probability $1/p$. Therefore, conditional on the complete history, a round decreases $r$ by at least one with probability at least $q_r=(1-p^{-r})/p$, and never increases it.

The remainder is identical to the geometric waiting-time analysis in Theorem~\ref{thm:explicit-exact}. In particular, the expected stopping time satisfies $\mathbb E[T_{\rm stop}]\leq\sum_{r=1}^n q_r^{-1}=pn+p\sum_{r=1}^n(p^r-1)^{-1}=pn+O_p(1)$, while $q_r\geq(p-1)/p^2$ for every $r\geq1$. The same concentration bound for sums of geometric random variables therefore gives Eq.~\eqref{eq:qudit-round-complexity}.

Once $r=0$, one has $D=\mathbb F_p^n$ and hence $D^\perp=\{0\}$, so all subsequent difference samples vanish. The transformed state is then a computational-basis state; one additional measurement determines its basis label, and inversion of the accumulated Clifford reconstructs the original stabilizer state.

\end{proof}

For $p=2$, the unknown slope has two possible values and the success probability $1/p$ reduces to $1/2$, recovering the qubit protocol and its leading round count $2n$. 

\end{document}